\documentclass[10pt]{article}
\usepackage{cite}
\usepackage{amsmath,amssymb,amsfonts}
\usepackage{algorithmic}

\usepackage{bbm}
\usepackage{textcomp}
\usepackage{lipsum}
\usepackage{float}
\usepackage{mathtools}
\usepackage{mathrsfs}
\usepackage{amssymb}
\usepackage{booktabs}
\usepackage{multirow}
\usepackage{subfig}
\usepackage{graphicx}
\usepackage{epstopdf}
\usepackage{breqn}
\usepackage{xcolor}
\usepackage{enumerate}
\usepackage{breqn}
\usepackage[normalem]{ulem}
\usepackage{arydshln}

\usepackage{enumitem}
\usepackage{amsmath}
\usepackage{cite}
\usepackage{refcount}
\usepackage{flushend}
 
\usepackage{amsthm}
\definecolor{DarkBlue}{rgb}{0.1,0.22,0.45}
\definecolor{DarkGreen}{rgb}{0,0.5,0}
\definecolor{DarkRed}{rgb}{0.66,0.01,0.1}

\usepackage[hyperfootnotes=false,colorlinks = true, linkcolor = black,citecolor = black]{hyperref}
\usepackage[capitalize]{cleveref}
\setlist{nolistsep}

\newtheorem{thm}{Theorem}
\newtheorem{deff}{Definition}
\newtheorem{rmk}{Remark}
\newtheorem{lem}{Lemma}

\newtheorem{proposition}{Proposition}
\newtheorem{assumption}{Assumption}

\newenvironment{newproofof}{\itshape{Proof of}}{\hfill$\square$}

\newcounter{alphasect}
\def\alphainsection{0}

\let\oldsection=\section
\def\section{%
	\ifnum\alphainsection=1%
	\addtocounter{alphasect}{1}
	\fi%
	\oldsection}%

\renewcommand\thesection{%
	\ifnum\alphainsection=1%
	\Alph{alphasect}
	\else%
	\arabic{section}
	\fi%
}%

\usepackage{fancyhdr}

\newcommand\shorttitle{EVENT-TRIGGERED DESIGN WITH GUARANTEED MINIMUM INTER-EVENT TIMES}
\newcommand\authors{MOHSEN GHODRAT AND HORACIO J MARQUEZ}

\begin{document}

\title{\normalsize\textbf{EVENT-TRIGGERED DESIGN WITH GUARANTEED MINIMUM INTER-EVENT TIMES AND $\mathcal{L}_p$ PERFORMANCE}}
                                            
                                            \author{\small MOHSEN GHODRAT AND HORACIO J MARQUEZ 
\thanks{The authors are with the Department of Electrical and Computer Engineering, University of Alberta, Edmonton, AB T6G 2V4, Canada (e-mail:
                                            		ghodrat@ualberta.ca; marquez@ece.ualberta.ca)}}
                                            \date{}

                                            \maketitle
                                            
\begin{abstract}                          
\small\baselineskip=9pt 
In an event-based scenario, the system decides when to update the actuators based on a real time triggering condition on the measured signals. This condition can be defined in various forms and varies depending on the system properties and design problem. This paper proposes a framework to design the triggering condition while keeping $\mathcal{L}_p$ performance within desired limits. 
Our general framework captures several existing state-based triggering rules as a special case, and  can achieve the performance objectives while reducing transmissions. Indeed, this general structure is shown to enlarge the minimum inter-event time by a specified amount, for a desired period of time. Numerical examples suggest that the proposed mechanism effectively enlarges the average sampling time. 
\end{abstract}


\section{Introduction}
{I}{t} is nowadays well accepted that event-triggered control provides a competitive alternative to traditional time-triggered control (\cite{tabuada,Lunze}), offering
performance very similar to classical controls while reducing the  transmission of information between system components.
Event triggered control was pioneered by \cite{astrom_stochastic} and has lead to extensive research including formal stability analysis (\cite{PID,tabuada, Heemels_periodic,Postuyan_GFW,Girard}) and the references therein, and performance in its various forms (\cite{L_infty_gain,lemmonlinearL_2fullstate,L_2_selftriggered, passive_input_output,passive_delay,Dolk_Lp,Dolk_LP_ieee}).   

Two important aspects of an event-triggered control are (i) the design should satisfy some form of closed-loop performance, and (ii) should guarantee that the execution times have enough separation to avoid excessive sampling. 
This second point is critical to any event design. Note that reducing communication between plant and controller is, in fact,  the primary motivation behind event-based methods.  However, since the execution times depend on the occurrence of a new event, the triggering rule has to be designed in a way to avoid excessive triggering, particularly the existence of an accumulation point in which an infinite number of events are generated in finite-time (also known as Zeno phenomenon).
In this regard, most event-triggered laws define a threshold using the norm of a measured signal, typically, the state. Examples include \cite{tabuada,L_2_selftriggered,Girard,Postuyan_GFW}. Although this type of scheme has seen countless of successful applications and has provided an important place in the literature, it is, however, not free of limitations. 
Indeed, in \cite{tabuada}, the author designed a triggering rule departing from a continuous-time closed-loop input-to-state stable (ISS) system (with respect to measurement error) to achieve an event-triggered closed-loop stable system, with Zeno-free behaviour. Similar rules can also guarantee other desirable performance measures such as 
$\mathcal{L}_2$ input-output bounds ({\it e.g.,} see \cite{lemmonlinearL_2fullstate,L_2_selftriggered}). However, it was recently shown in \cite{Dolk_Lp} that in the presence of disturbance or sensor noise, static event-triggering rules defined in terms of solely the state vector norm cannot guarantee positive minimum inter-event time (MIET), thus becoming vulnerable to Zeno behaviour. The same issue may be encountered when dealing with dynamic triggering rules (\cite{Postuyan_GFW,Girard}), or output-based triggering rules (\cite{L_infty_gain}).


The observations above suggest that incorporating disturbance in event design is nontrivial. 
One way to address this issue is the so-called time-regularization method, recently proposed in \cite{Mahmoud_Abdolrahim,Dolk_LP_ieee}. The core idea is to enforce a dwell-time between sampling times, inspired by classical periodic sampling ({\it e.g.,} see \cite{explicit-nesic}). See also \cite{Chopra,Forni,Mazo_Tabuada} for a different approach. 
A common feature of all these studies is that the triggering condition (TC) is checked only after some specific time has elapsed since the most recent triggering instant and hence the positiveness of MIET is satisfied by construction. 
One drawback of this technique is that it mostly reduces to periodic samplings whenever TC is static and state is near the origin, \cite{Dolk_Lp,event-separation}.

\emph{Statement of contributions:}
Our primary objective is to design a prescriptive triggered policy for nonlinear event-triggered systems (ETSs) such that an $\mathcal{L}_p$ performance measure is satisfied in the presence of exogenous disturbances. 
As opposed to the time-regularization approach, in this article we focus on ``pure" event-based TC referring to continuous monitoring of the triggering rule without imposing a pre-defined dwell-time. 
Following this method, the previously mentioned shortcoming associated with time-regularization approach is prevented at the expense of a non-trivial positive MIET problem. Instead, to exclude Zeno-behaviour we introduce a new dynamic parameter in the proposed TC.

The results in this work are indeed a generalization of our previous work, \cite{Me_submitted}, where we studied the local $\mathcal{L}_2$ stability of nonlinear ETSs under state-dependent disturbances. Here, we first relax the restriction on the admissible input space and allow disturbances to be any signal in $\mathcal{L}_p$ space, thereby extend the results of \cite{Me_submitted} to global (non-local) $\mathcal{L}_p$ type performance. This extension, however, produces non-trivial challenges regarding the Zeno-exclusion that will be discussed in detail. Secondly, we cast a general structure for event-triggering design which specifically captures the proposed event in \cite{Me_submitted}.

In summary, our design addresses the two main characteristics of ETSs mentioned earlier, in the following senses:

\emph{Performance.} We focus on $\mathcal{L}_p$-gain performance of ETSs. We consider a general class of nonlinear control-affine system and a pre-designed state feedback controller whose continuous implementation satisfy some $\mathcal{L}_p$-gain performance level $\mu$. We provide a constructive TC design algorithm to retain the new $\mathcal{L}_p$-performance level at some $\mu_d$. The results here extends the prior work \cite{Me_submitted}, where the disturbance is assumed to be bounded. 
See also  \cite{Dolk_LP_ieee,Mahmoud_Abdolrahim} for a different approach. In comparison to these references, we rely on less conservative set of assumptions and a different approach that lead to a different structure for TC design. In fact, the assumptions made in \cite{Dolk_LP_ieee,Mahmoud_Abdolrahim} require some sort of dissipativity property for ETS which, we believe, is too strong to be applied to the problem considered here. Moreover, in the dwell-time approach, $\mu_d-\mu$ is lower bounded by some function of the dwell-period thus limiting how close the event-triggered performance $\mu_d$ can be made to its continuous-time counterpart $\mu$. This limitation, however, does not exist in our proposed technique.

\emph{Reduced transmission.} We propose a general framework for constructing state-based dynamic TCs in which the transmissions between plant and controller is significantly reduced. Rather than a single result, our proposed dynamic TC contains design parameters that can be selected for specific purpose and covers several well-known forms (namely, \cite{tabuada,Postuyan_GFW,Girard,me_iet,Dolk_Lp,Dolk_LP_ieee,event-separation,Me_submitted}) as special cases. 
In addition, as in \cite{Me_submitted}, we show that the proposed TC has the advantage over the existing dynamic techniques that a lower bound on inter-event times extension can be set.
Recent studies in \cite{Mahmoud_Abdolrahim,Dolk_LP_ieee} consider the more realistic output feedback event-triggered control problem. Our focus is on the state feedback case. While state feedback has its practical limitations, the advantage is that the result is less conservative than the output feedback case. An important feature of our design is that the introduced dynamic parameters are applicable to the available TCs without violating $\mathcal{L}_p$ or ISS performance. The resulting triggering rule thus enjoys improved sampling times. 


\textit{Notation.} 
$\mathbb{R}$ (resp., $\mathbb{Z}$) represents the field of real numbers (resp., set of integers). 
$\mathbb{R}^+_0$, $\mathbb{Z}^+_0$, $\mathbb{R}^+$, $\mathbb{Z}^+$ are the sets of nonnegative and positive elements of $\mathbb{R}$ and $\mathbb{Z}$. $\|x\|$ denotes the Euclidean norm of vector $x\in\mathbb{R}^n$.
$\mathcal{L}_p^n$ is the space of measurable n-dimensional signals $w$ with bounded $p$-norm defined as $(\int_{t_0}^{\infty}\|w(t)\|^pdt)^{\frac{1}{p}}$. The $\infty$-norm of $w$ is denoted by ${\|w\|}_{\infty}=\text{ess}\sup\{\|w(t)\| : t\geq t_0 \}$.
A function $f:\mathbb{R}^{n} \mapsto\mathbb{R}^p$ is said to be locally Lipschitz-continuous in an open set $D$, if for each $z\in B$ there exist $L_f,r\in\mathbb{R}^+$ such that $\|f({{x}})-f(\tilde{{x}})\|\leq L_f\|{{x}}-{\tilde{{x}}}\|$ for all ${{x}},{\tilde{{x}}}\in \{y\in D:\|y-z\|<r\}$. 
A sequence $\{x_i:i\in\mathbb{Z}^+_0\}$ is  uniformly isolated iff there exists some $r\in\mathbb{R}^+$ so that $|x_i-x_j|>r$ for any $i,j\in\mathbb{Z}^+_0$ with $i\neq j$. 

\section{Event triggered control system} \label{preliminaries}
Consider the nonlinear system of the following form \footnote{The assumed affine structure is not critical and can be relaxed at the expense of obtaining a more conservative triggering condition in the sense that the triggering threshold would be reached sooner, see Remark \ref{affine}.}:
\begin{eqnarray}\label{eq sys_0}
\dot{\xi}=f(\xi,d)+g(\xi)u,~~~
z=h(\xi,d),
\end{eqnarray}
where $\xi\in\mathbb{R}^n$, $u\in\mathbb{R}^m$, $d\in \mathcal{L}_{p}^q$, $z\in\mathbb{R}^s$ represent the state, control input, exogenous disturbance and measured output. 
The functions $f$, $g$ and $h$ are locally Lipschitz-continuous and $f(0,0)=0$, $h(0,0)=0$ so that $\xi=0$ is an equilibrium point of zero-input system. We will assume the state $\xi$ evolves from initial conditions $\xi_0=\xi(t_0)$ on an open subset of $\mathbb{R}^n$ containing the origin.
System (\ref{eq sys_0}) is said to be finite gain $\mathcal{L}_p$-stable and has an $\mathcal{L}_p$-gain $\leq\mu$ if there exist real numbers $\eta,t_0,T$, $\mu>0$, $p\geq 1$ and positive semi-definite function $\beta$ such that for any $T>t_0$, any $d\in\mathcal{L}_{p}^q$ and any $\xi_0\in\mathbb{R}^n$
	\begin{eqnarray}\label{practical eq}
	\int_{t_0}^{T}{\|z(s)\|^p}ds\leq\mu^p\int_{t_0}^{T}{\|d(s)\|^p}ds+\beta(\xi_0)+\eta.
	\end{eqnarray}
We assume plant and controller communicate aperiodically through a digital network and in an event-based manner. The event-triggered problem established in this paper relies on the emulation of the analog design and consists of two steps: 

First, we assume continuous data transmission between plant and a full information controller $u=\gamma(\xi)$, where $\gamma$ is locally Lipschitz-continuous. The resulting continuous-time plant is then given by
\begin{eqnarray}\label{eq sys_1}
\dot{\xi} ={f_c}(\xi,d),~~~ z=h(\xi,d), 
\end{eqnarray}
where ${f_c}(\xi,d):=f(\xi,d)+g(\xi)\gamma(\xi)$. It is then assumed that the controller renders the closed-loop (\ref{eq sys_1}) finite gain $\mathcal{L}_p$-stable with disturbance attenuation level $\mu$.

Second, the communication between plant and controller occurs at the instants belong to the set ${\{t_k:{k\in\mathbb{K}}\}}$, where $\mathbb{K}= \{0,1,2,\ldots,K\}$. The sampling sequence is a monotone increasing set, starting at $t_0$ and implicitly defined through a triggering rule. The actuator signal is held constant between events using a hold device $u(t)=u(t_k)$, $t\in[t_k,t_{k+1})$ where $t_{K+1}=\infty$ when $K$ is finite.
The proposed TC is continuously monitored and once it is satisfied, the updated state is forwarded to the controller which computes the new control signal and send it to the actuator instantaneously. More specifically, let $t_k$ be the most recent sampling instant and TC be satisfied at some $\varpi_{k+1}>t_{k}$. Then the new control signal applied through the actuator at $t_{k+1}=\varpi_{k+1}^+$ and hence $u(t_{k+1})=\gamma(\xi(t_{k+1}))$. Let $\varepsilon(t):=\xi(t_k)-\xi(t)$ represent the sampling error for $t\in[t_k,t_{k+1})$. $\varepsilon(t)$ is then a right-continuous signal with zero value at $t_{k}$. In our analysis we neglect practical issues such as transmission and computation delays, however, they can be readily addressed following the approach introduced in \cite{tabuada}. The resulting closed-loop ETS is then described by 
\begin{eqnarray}\label{eq sys}
\begin{cases}
\dot{\xi}={f}_s(\xi,\varepsilon,d),~~~
 z=h(\xi,d),\\  t_{k+1}=\inf \big\{t\in\mathbb{R}:t>t_k \wedge \Phi(t^-)=0\big\},
\end{cases}\hspace*{-1.6em}
\end{eqnarray} 
where ${f}_s(\xi,\varepsilon,d):=f(\xi,d)+g(\xi)\gamma(\xi+\varepsilon)$ and $\Phi(t)$ is the TC to be designed.

Assuming the existence of an $\mathcal{L}_p$-stabilizing controller for (\ref{eq sys_1}), our main interest is to design an event-triggered mechanism (ETM) that retains this input-output property of the network-free design for the resulting event-based plant; perhaps with a worse disturbance attenuation level.
The proposed ETM shall (1) exclude the Zeno behaviour and (2) serve as a general platform for TC design in event-based problems. 

\begin{rmk}\label{remark-future}
	Let us write the system dynamic as $\dot{\xi}=f(\xi,d)+\sum\nolimits_{j=1}^{m}g_j(\xi)u_j$ 
	where $f,g_j:\mathbb{R}^n\rightarrow \mathbb{R}^n$, $g=(g_1,\ldots,g_m)$ and $\gamma=(\gamma_1,\ldots,\gamma_m)^{\mathsf{T}}$. Recent studies \cite{decentralized-Mazo,Wang-distributed} consider the interesting scenarios that controller is (i) \text{distributed}: $u_j(t)=\gamma_j(\hat{\xi}(t))$, $j=1,\ldots, m$, and (ii) \text{decentralized}: $\hat{\xi}_i(t)=\xi_i(t^i_{r_i})$, $t\in[t^i_{r_i},t^{i}_{r_{i+1}})$, $i=1,\ldots,n$,
	implying that the distributed controllers $u_j=\gamma_j(\hat{\xi})$ utilizes the full state vector measured through $n$ independent sensors to construct the control signal. Studying the results of the current paper under these assumptions is left as a future work.
\end{rmk}

\section{Event-triggering mechanism} \label{section_ETM}
In this section, we introduce a general structure to design $\Phi$ so that ETS (\ref{eq sys}) has $\mathcal{L}_p$-gain $\leq \mu_d$. 
Consider the following TC structure:
\begin{eqnarray}\label{trig_cond}
\Phi(t) := \bar{k} \varphi(\xi(t),\varepsilon(t)) -\textstyle{\sum} 
_{i=1}^{2}k_i\phi_i(t)=0
\end{eqnarray}
where $k_1, k_2 >0$, and the  dynamic variables ${\phi_1}$, ${\phi_2}$ and function $\varphi$ are to be designed. Furthermore, we will assume $\bar{k}=1$ unless otherwise stated.
We start with designing $\varphi$, for which the following assumption is required. 
\begin{assumption}\label{assumption_ISS}
	There exist positive definite, radially unbounded functions $V_s$, $V_c$, positive constants $\mu$, $c_i$, $\bar{c}_i$ $i\in\{1,2,3\}$ and some $p\in[1,\infty)$ satisfying \vspace*{.1em}
	\begin{itemize}
		\item [(i)] $\nabla{V_s}(\xi)\cdot f_s(\xi,\varepsilon,d)\leq -c_1\|\xi\|^p+c_2\|\varepsilon\|^p+c_3\|d\|^p$,\vspace*{.1em}
		\item [(ii)] $\nabla{V_c}(\xi)\cdot f_c(\xi,d)\leq\mu^p\|d\|^p-\|z\|^p$,\vspace*{.1em}
		\item [(iii)] $V_s(\xi)\leq \bar{c}_1\|\xi\|^p$, $V_c(\xi)\leq\bar{c}_2\|\xi\|^p$, $\|\nabla V_c(\xi)\|\leq \bar{c}_3\|\xi\|^{p-1}$.
	\end{itemize}
\end{assumption}
\begin{rmk}
	Assumption \ref{assumption_ISS}(i) implies that system (\ref{eq sys}) is ISS with respect to the inputs $\varepsilon$, $d$. Also Assumption \ref{assumption_ISS}(ii) implies that $u=\gamma(\xi)$ renders the continuous-time system (\ref{eq sys_1}) finite gain $\mathcal{L}_p$-stable with $\mathcal{L}_p$-gain $\leq \mu$.
\end{rmk}
The function $\varphi$ is assumed to have the following form
\begin{eqnarray} \label{varphi}
\varphi(\xi,\varepsilon)  = \varphi_1(\xi)+\varphi_2(\varepsilon)+\varphi_3(\xi,\varepsilon),
\end{eqnarray}
where $\varphi_1(r)  =  - c_1\sigma\|r\|^p$, $\varphi_2(r) = c_2\|r\|^p$,
\begin{eqnarray}\label{phi_3}
\varphi_3(r,s)  = \nabla V_{c,\lambda}(r)\cdot g(r)(\gamma(r+s)-\gamma(r)) 
\end{eqnarray} 
and $\sigma< 1$, $p\in[1,\infty)$, $V_{c,\lambda}(r)=\lambda V_c(r)$ for some $\lambda\in\mathbb{R}^+$. 
We then continue with the design of $\phi_1$, $\phi_2$; dynamic parameters serve to enlarge the inter-event times and guarantee the event-separation property for ETS (\ref{eq sys}). Consider the equations below for $t\in[t_k,t_{k+1})$
\begin{subequations}\label{phi_eta_Ti}
\begin{eqnarray}
\frac{d}{dt}\begin{pmatrix}
{\phi}_1\\{\phi}_2
\end{pmatrix}+\begin{pmatrix}
\alpha_1({\phi_1})-k_2{\phi_2}\\ \alpha_2(\phi_2)
\end{pmatrix}=\begin{pmatrix}
-\varphi\\ \bar{\varphi}
\end{pmatrix},\label{phi}~~~~~
\\
{\bar{\varphi}(t)}=
\begin{cases}
\alpha_2(\bar{\delta}),~~~~~~~~~~~~~~~~t\in[t_k,\hat{t}_k),\\
\dot{\delta}_k(t)+\alpha_2(\delta_k(t)),~~ t\in[\hat{t}_k,t_{k+1}),
\end{cases}\label{Delta}
\end{eqnarray}
\end{subequations}
where $\bar{\delta}$ is a positive constant and $\delta_k$ is a positive, bounded and piecewise differentiable function defined over $[\hat{t}_k,t_{k+1})$ and satisfies $\sum_{k} \int_{\hat{t}_k}^{t_{k+1}} \delta_k(\tau)d\tau\leq \theta_1$ for some positive $\theta_1$. Also $\hat{t}_k=t_k+\hat{\tau}$, where $\hat{\tau}$ is a positive parameter and will be designed in the sequel. Note that function $\bar{\varphi}$ is defined such that $\bar{\delta}$ (resp., $\delta_k(t)$) is a solution of $\phi_2$ in (\ref{phi}) over $[t_k,\hat{t}_k)$ (resp., $[\hat{t}_k,t_{k+1})$). 
Moreover $\alpha_2$ is an arbitrary class-$\mathcal{K}_\infty$ function and $\alpha_1 \in \mathcal{K}_\infty$ is designed based on the following assumption.
\begin{assumption}\label{ass_alpha}
	$\alpha_1(r)\geq \nu r$ where $\nu=c_1({1-\sigma})/({\bar{c}_1+\bar{c}_2})$.
\end{assumption}
To solve (\ref{phi}), (\ref{Delta}) the following initial values are assumed
\begin{eqnarray}\label{IC}
{\phi_1}(t_k)=r_k,~{\phi_1}(\hat{t}_k)={\hat{r}_k},~
\phi_2(t_k)=s_k,~\phi_2(\hat{t}_k)={\hat{s}_k},
\end{eqnarray}
where $r_k$, $\hat{r}_k$, $s_k$, $\hat{s}_k$ are non-negative real numbers and are designed based on the following assumption.
\begin{assumption}\label{ass_impose}
	$r_k$ and $\hat{r}_k$ are chosen from sequences with convergent series, {\it i.e.,} there exist finite numbers ${\theta_2},\theta_3\in\mathbb{R}^+$ so that $\sum\nolimits_{k}r_k \leq {\theta_2}$, $\sum\nolimits_{k}\hat{r}_k \leq \theta_3$. Moreover, $s_k$ and $\hat{s}_k$ satisfy $s_k\geq \bar{\delta}$ and $\hat{s}_k= \delta_k(\hat{t}_k)$. 
\end{assumption}
Dynamic rules have been previously studied in  \cite{Postuyan_GFW,Girard}. The variable ${\phi_1}$ in (\ref{trig_cond}) which satisfies the differential equation (\ref{phi}), plays the role of dynamic parameter introduced in the above references. In the present work, we introduce an additional dynamic variable ${\phi_2}$; while both $\phi_1$, $\phi_2$ serve to extend the inter-event times, $\phi_2$ plays the fundamental role of guaranteeing event separation property for ETS (\ref{eq sys}). In particular, $\phi_2$ introduces two design parameter, $\bar{\delta}$, $\delta_k$. The former, is inspired by the idea of mixing triggering condition in \cite{event-separation} and is intended to rule out Zeno-behaviour. The latter, on the other hand, is a generalization of time-decaying thresholds and is mainly used to move from practical asymptotic stability (under constant threshold) to asymptotic stability, \cite{decentralized-Mazo,Seyboth}. Designing these parameters together with $\hat{t}_k$ which decides the duration over which each parameter is effective, lead to non-trivial challenges that have to be carefully carried out. Based on the above observations, the proposed TC (\ref{trig_cond}) unifies (i) dynamic TC \cite{Girard} through the presence of $\phi_1$, (ii) time-varying threshold \cite{decentralized-Mazo} through presence of $\delta_k$ and (iii) mixed triggering \cite{event-separation} through presence of $\bar{\delta}$.
\begin{proposition}\label{prop_phi_positive}
	Under TC (\ref{trig_cond}) and Assumption \ref{ass_impose}, ${\phi_1}(t),{\phi_2}(t)\geq0$ for all $t\geq t_0$. In detail, ${\phi_2}(t)\geq \bar{\delta}$ for $t\in[t_k,\hat{t}_k)$, ${\phi_2}(t)=\delta_k(t)$ for $t\in[\hat{t}_k,t_{k+1})$.
\end{proposition}
Proposition \ref{prop_phi_positive}, whose proof is provided in the Appendix, illustrates the previous claim that ${\phi_1}$, ${\phi_2}$ enlarge the inter-event times. In fact, in absence of $\phi_1$, $\phi_2$ triggering occurs when $\varphi(\xi,\varepsilon)=0$. However, the positiveness of ${\phi_1}$, ${\phi_2}$ postpones the triggering to occur when $\varphi(\xi,\varepsilon)=k_1\phi_1+k_2\phi_2$.
To finish the design, it remains to define $\hat{\tau}$. Let us start with the following lemma whose proof is given in Appendix. 
\begin{lem}\label{lem_state_bound}
	Under Assumptions \ref{assumption_ISS}-\ref{ass_impose} and if the control signal is updated under the triggering rule (\ref{trig_cond}), all the trajectories of the ETS (\ref{eq sys}) starting from $\mathscr{B}_{\rho}$ will remain in $\mathscr{B}_{\bar{\rho}}$, where
	\begin{eqnarray}
	\bar{\rho}=\displaystyle\max\Big{\{}\|\xi\|:V_s(\xi)+V_{c,\lambda}(\xi)\leq V_s(\xi_0)+V_{c,\lambda}(\xi_0)+~~~~\nonumber\\ ~\frac{1}{\nu}( \lambda\mu_d^p{\|d\|}_\infty^p+{k_2\|{\phi_2}\|}_{\infty})+{\theta_2}+\theta_3, \xi,\xi_0\in\mathbb{R}^n,\|\xi_0\|\leq \rho\Big{\}}.\nonumber
	\end{eqnarray}
\end{lem}
Since  ${\|{\phi_2}\|}_{\infty}$ is limited by $\smash{\displaystyle{\max}\{s_k,{\|\delta_k\|}_{\infty}:k\in\mathbb{K}\}}$ and hence is bounded, Lemma \ref{lem_state_bound} suggests that the trajectories of the ETS (\ref{eq sys}) are bounded by a non-decreasing function of $\|\xi_0\|$ and ${\|d\|}_{\infty}$. Next lemma employs the Lipschitz property of $f$, $g$, $\gamma$ to provide an upper bound on the norm of state dynamics.
\begin{lem}\label{lem_lip}
	With the same conditions as in Lemma \ref{lem_state_bound}, there exist $\lambda_i=\lambda_i({\|\xi_0\|},{\|d\|}_\infty)$, $i\in\{1,2,3\}$, non-decreasing on their arguments, so that 
	\begin{eqnarray}
	\| \dot{\xi} \| \leq \lambda_1\|\xi\|+\lambda_2\|\varepsilon\|+\lambda_3\|d\|.\nonumber
	\end{eqnarray}
\end{lem}
\begin{proof}[Sketch of the proof]
	One can apply the Lipschitz property of functions $f$, $g$, $\gamma$ to get $ \|\dot{\xi} -\dot{\tilde{\xi}}\|  \leq \lambda_1\|\xi-\tilde{\xi}\|+\lambda_2\|\varepsilon-\tilde{\varepsilon}\|+\lambda_3\|d-\tilde{d}\|$
	where $\lambda_i$'s are functions of ${\|d\|}_\infty$ and $\bar{\rho}$. The result then follows by applying Lemma \ref{lem_state_bound}. 
\end{proof}
\begin{rmk}
	As Lemma \ref{lem_lip} suggests, since the Lipschitz properties of functions $f$, $g$ and $\gamma$ are only \emph{local}, the Lipschitz coefficients $\lambda_i$ are bounded provided $\|\xi_0\|<\infty$, ${\|d\|}_{\infty}<\infty$.
\end{rmk} 
Let us define for $i\in\{1,3\}$
	\begin{eqnarray}
	\tau_i:=\sup\Big{\{}t\in\mathbb{R}^+_0:\lambda_i^p \psi(t,\lambda_2)<\frac{B_i}{c}  \Big{\}},\nonumber
	\end{eqnarray}
	where $B_1=c_1\sigma-\frac{(\bar{c}_3\lambda)^q}{q}$, $B_3=\lambda(\mu_d^p-\mu^p)-c_3$, $c=c_2+\frac{\lambda_2^p}{p}$ and
	$\psi(t,\lambda_2)=\frac{2^{2p}(p-1)^{p-1}}{\lambda_2^{p}p^{p}}(e^{\frac{\lambda_2p}{2(p-1)}t}-1)^{p-1}(e^{\frac{\lambda_2p}{2}t}-1).$ 
	Then, we define $\hat{\tau}$ as
\begin{eqnarray}\label{tau_hat}
\hat{\tau}=\min\{\tau_1,\tau_3\}.
\end{eqnarray}
Later in Lemma \ref{semi_global} we will see that $\hat{\tau}>0$ serves to isolate of triggering instants. Moreover, $\tau_1$ (resp., $\tau_3$) is the elapsed time since the most recent triggering instant so that sampling error grows without violating stability (resp., desired $\mathcal{L}_p$ bound) of the ETS (\ref{eq sys}) (see the proof of Theorem \ref{thm_main}). 
	In addition, the definitions of $\hat{\tau}$, $\tau_3$ infers that $\mu_d^p-\mu^p$ is lower bounded by $\psi(\hat{\tau},\lambda_2)$. Since as shown in Section \ref{compare}, $\hat{\tau}$ is a candidate for dwell-period, the dwell-time method places a lower bound on the deviation of continuous-time and event-based performances. This restriction does not exist in our approach. 

\begin{rmk}\label{lambda}
To design $\lambda$ one has to consider the restriction of having a positive $\hat{\tau}$. 
$\hat{\tau}>0$ necessitates $\tau_1$, $\tau_3$ and hence $B_1$, $B_3$ to be positive. This gives the restriction on $\lambda$ as $\lambda < {\bar{c}_3^{-1}}{{(c_1\sigma q)}^{\frac{1}{q}}}$ and $\lambda>{c_3}({\mu_d^p-\mu^p})^{-1}$. The later condition implies $\mu_d>\mu$, {\it i.e.,} the $\mathcal{L}_p$-stability of ETS (\ref{eq sys}) is achieved at the expense of a larger rejection level. However, to minimize $\mu_d$, we may choose $c_3$ small enough by \emph{scaling} Lyapunov function $V_s$ in Assumption \ref{assumption_ISS} (refer to example section for more details). Obviously, one has to replace $c_i$, $i\in\{1,2,3\}$ by the corresponding scaled values in all of the discussions. 
\end{rmk}
\section{Main results}\label{preliminary result}
\subsection{Uniform isolation of triggering instants}
%
One of the difficulties encountered in event-based control systems is undesirable Zeno behaviour which happens when an infinite number of triggerings occur over a finite interval. This is even more challenging when the system of interest is exposed to exogenous disturbances or sensor noise, since in this case the sampling error is also driven by the disturbance/noise. As an example, while Zeno behavior is excluded in \cite{tabuada} for disturbance free systems, the same does not necessarily hold in presence of disturbance (see \cite{event-separation} for further discussion).
In the sequel, we show that under TC (\ref{trig_cond}), the ETS (\ref{eq sys}) satisfies the following \emph{robust event-separation} property defined in \cite{event-separation}.
\begin{deff}\label{def-separation}
	Let $\tau_{m} =\inf\{t_{k+1}-t_k : k\in\mathbb{K}\}$ be the MIET. ETS (\ref{eq sys}) has the robust semi-global event-separation property if there exists $\epsilon\in \mathbb{R}^+$ so that for any compact set $\Xi\subset\mathbb{R}^n$, $\inf\{\tau_m : \xi_0\in\Xi ,{\|d\|}_\infty\leq \epsilon \}>0$.
\end{deff}
According to Definition \ref{def-separation}, an event-based system has the robust semi-global event-separation property if the sequence of sampling times 
$\{t_k:k\in\mathbb{K}\}$ is a uniformly isolated set provided that $\xi_0\in\Xi$ and ${\|d\|}_\infty\leq \epsilon$.

\begin{lem}\label{semi_global}
	Under Assumptions \ref{assumption_ISS}, \ref{ass_alpha}, \ref{ass_impose} and ETM (\ref{trig_cond})-(\ref{tau_hat}), the ETS (\ref{eq sys}) has the robust semi-global event-separation property. In detail,
	\begin{eqnarray}
	\tau_{m}=\displaystyle\min\{\tau^*(1),\hat{\tau}\},\nonumber
	\end{eqnarray}
	where $m_1=(\frac{B_1}{c})^{\frac{1}{p}}$, $m_2=(\frac{k_2 \bar{\delta}}{c})^{\frac{1}{p}}$, $\kappa{=}\max\Big{\{}\frac{2\lambda_1}{m_1},\frac{2\lambda_3 \epsilon}{m_2}\Big{\}}$ and
	\begin{eqnarray}\label{chi}
	\tau^*(\chi)=\begin{cases} 
	\frac{1}{\lambda_2-m_1\frac{\kappa}{2}}\ln(\frac{\kappa+\lambda_2\chi	}{\kappa(1+m_1\frac{\chi}{2})}),~~~~\kappa\neq\frac{2\lambda_2}{m_1},\\
	\frac{m_1{\chi}}{\lambda_2(2+m_1{\chi})},~~~~~~~~~~~~~~~~~~\kappa=\frac{2\lambda_2}{m_1}.
	\end{cases} 
	\end{eqnarray}
\end{lem}
To prove Lemma \ref{semi_global} we report here two useful inequalities.
\begin{lem}\label{lem_1}
	For any $p,q\geq 1$ with $\frac{1}{p}+\frac{1}{q}=1$ and any $r>0$
	\begin{eqnarray}
	(i)~\|x+y\|^r\leq 2^r\|x\|^r+2^r\|y\|^r ~~~~~~~~~~~~~~~~~~~~~~~~~~~~~~~~~~~~~\nonumber\\
	(ii)\int_{\mathcal{T}}\|x(\tau)y(\tau)\| d\tau \leq \Big{(}\int_{\mathcal{T}}\|x(\tau)\|^p d\tau\Big{)}^{\frac{1}{p}} \Big{(}\int_{\mathcal{T}}\|y(\tau)\|^q d\tau\Big{)}^{\frac{1}{q}}\nonumber.~
	\end{eqnarray}
\end{lem}
\begin{newproofof}
	\textit{Lemma \ref{semi_global}.}
	We aim to modify (\ref{trig_cond}) to obtain a more conservative TC (in the sense that triggering threshold would be reached sooner) since such a TC gives rise to a lower bound on MIET. To begin, we first make the use of Proposition \ref{prop_phi_positive} which implies $\phi_1\geq 0$, ${\phi_2}(t)\geq \bar{\delta}$ for $t\in[t_k,\hat{t}_k)$ and hence modify (\ref{trig_cond}) as $\varphi=k_2\bar{\delta}$ in this interval. Note that we will assume $t_{k+1}\leq \hat{t}_k$ since otherwise $\tau_m=\hat{\tau}$ and the event-separation property holds trivially.
	From the inequality given in the sketch of proof of Lemma \ref{lem_lip} with $\tilde{\xi}=\xi$, $\tilde{\varepsilon}=0$, $\tilde{d}=d$, one can conclude $\|g(\xi)(\gamma(\xi+\varepsilon)-\gamma(\xi))\|\leq \lambda_2 \|\varepsilon\|$ and hence modify condition $\varphi=k_2\bar{\delta}$ as 
	\begin{eqnarray}
	c_2\|\varepsilon\|^p+\lambda_2\|\nabla V_{c,\lambda}(\xi)\|\|\varepsilon\|=c_1\sigma\|\xi\|^p+k_2\bar{\delta}.\nonumber
	\end{eqnarray}
	Next, from Lemma \ref{lem_1}(ii) and Assumption \ref{assumption_ISS}(iii), we obtain $c \|\varepsilon\|^p =B_1\|\xi\|^p+k_2\bar{\delta}$. Finally, Lemma \ref{lem_1}(i) suggests 
	\begin{eqnarray}
	\Big{(}\frac{{B_1}^{\frac{1}{p}}}{2}\|\xi\|+\frac{({k_2\bar{\delta}})^{\frac{1}{p}}}{2}\Big{)}^p\leq B_1\|\xi\|^p+k_2\bar{\delta},\nonumber
	\end{eqnarray}
	and hence, the desired modification of (\ref{trig_cond}) is obtained as 
	\begin{eqnarray}\label{m1m2}
	2\|\varepsilon\|=m_1{\|\xi\|}+m_2.
	\end{eqnarray}
	Define $\chi:=2\|\varepsilon\|/(m_1{\|\xi\|}+m_2)$, it can be concluded that
	\begin{eqnarray}
	\dot{\chi}\leq \Big{(}1+m_1\frac{\chi}{2}\Big{)}\Big{(}\frac{2\|\dot{\xi}\|}{m_1{\|\xi\|}+m_2}\Big{)}\leq \Big{(}1+m_1\frac{\chi}{2}\Big{)}\Big{(}\kappa+\lambda_2\chi\Big{)}\nonumber
	\end{eqnarray}
	where Lemma \ref{lem_lip} is used to obtain the last inequality. Therefore, $\tau^*(\chi)=t-t_k$ can be obtained as in (\ref{chi}) by solving
	\begin{eqnarray}
	\displaystyle \int_{t_k}^{t}d\tau=\displaystyle \int_{\chi(t_k)=0}^{\chi(t)}{\big{(}1+m_1\frac{\ell}{2}\big{)}^{-1}\big{(}\kappa+\lambda_2\ell\big{)}^{-1}}{d\ell}.\nonumber
	\end{eqnarray}
	Event rule (\ref{m1m2}) suggests that triggering occurs when $\chi=1$, thus $t_{k+1}=t_k+\tau^*(1)$. In addition, (\ref{chi}) implies that $\tau^*(1)$ is strictly nonzero since for $\xi_0\in\Xi$ and ${\|d\|}_\infty\leq \epsilon$, Lemma \ref{lem_lip} suggests that $\lambda_1$, $\lambda_2$, $\lambda_3$, and hence $\kappa$ are bounded. 
	The result then follows from definition of $\tau_{m}$ and positiveness of $\hat{\tau}$. 
\end{newproofof}
\subsection{Comparison with the existing strategies}\label{compare}
In this subsection we study several popular existing ETMs that can be extracted as special cases of (\ref{trig_cond})-(\ref{Delta}). We emphasize that the design criteria in these references is not the same so our comparison is merely based on the structure of the triggering rule
with no reference to the relative merits or performance in each design, simply because there seems to be no fair way or value in such comparison.
Moreover, since some of these works focus on output feedback, in our comparisons we assume the measurable output to be the full state vector.

Our proposed ETM is \emph{dynamic} due to the existence of the dynamic variable $\phi_1$. See \cite{Girard,Postuyan_GFW} for discussions regarding the effect of this variable. To the best of our knowledge, the parameter $\phi_2$ has not been introduced before. Thus, we provide the following observations regarding $\phi_2$.

(i) The inter-event expansion that originates from $\phi_2$ can be quantified for a desired period of time, or a desired number of trigger instants (see \cite{Me_submitted}). 
(ii) As shown in \cite{Me_submitted} through several examples, ${\phi_2}$ serves to avoid redundant samplings when the norm of state is close to $0$. This is important since as a primary pitfall, triggering rules based on the norm of the state tend to increase triggering as the state approaches the origin. 
(iii) The primary functionality of ${\phi_2}$ is to exclude Zeno behaviour as suggested by the proof of Lemma \ref{semi_global}. 
(iv) While the approach in the present article is considered purely event-based, an appropriate choice of parameters in the dynamics of ${\phi_2}$ enables the TC (\ref{trig_cond}) to capture the time-regularization strategies. 
We conclude by extracting several triggering rules proposed in the literature from (\ref{trig_cond}). Note that $\bar{k}=1$ unless otherwise stated.

$\bullet$ \cite{Me_submitted}: For $k_1=0$ and $s_k=\bar{\delta}$, TC (\ref{trig_cond}) reduces to the one proposed in \cite{Me_submitted}. 

In the rest of our comparisons we assume $\varphi_3 = 0$ in (\ref{varphi}).

$\bullet$ \cite{event-separation}:
	For $k_1=0$ and $\delta_k(t)=s_k=\hat{s}_k=\bar{\delta}$ we obtain ${\phi_2}=\bar{\delta}$. Hence, the TC becomes $\varphi(\xi,\varepsilon)=k_2\bar{\delta}$. 
$\bullet$ \cite{Girard}:
	Take $k_2=0$, (\ref{trig_cond}), (\ref{phi}) reduce to $\varphi(\xi,\varepsilon)=k_1 {\phi_1}$, $\dot{\phi}_1+\alpha_1({\phi_1}) =-\varphi$.
$\bullet$ \cite{me_iet}: 	
	Taking $k_2=0$, $\bar{k}=0$ and $\alpha_1(r)=0$ for any $r$, (\ref{trig_cond}) reduces to ${\phi_1}=0$, where ${\phi_1}(t)=-\int_{t_k}^{t} \varphi(\xi(s),\varepsilon(s))ds$, {\it i.e.,} the integral-based TC. 
$\bullet$ \cite{tabuada}: 
	Substitute $k_1=k_2=0$ in (\ref{trig_cond}) one can extract the TC $\varphi=0$. 
$\bullet$ \cite{Dolk_LP_ieee,Dolk_Lp}:
Define $\hat{t}_k=t_k+\tau_{m}$ where $\tau_{m}= \displaystyle\min\{\tau^*(1),\hat{\tau}\}$. This guarantees no triggering of the control task occurs over $[t_k,\hat{t}_k)$. Then, if we set $k_2=0$ and $\bar{k}=0$, $\Phi(t)=0$ reduces to $\phi_1(t)=0$, thereby $t_{k+1}$ in (\ref{eq sys}) can be written in a time-regularization fashion as $t_{k+1}=\inf \{t\in\mathbb{R}:t>t_k+\tau_{m} \bigwedge  {\phi_1}(t^-)=0\}$ where $\dot{\phi}_1=-\varphi$ by setting $k_2=0$ and $\alpha_1(r)=0$ for any $r$ in (\ref{phi}). 
 $\bullet$ \cite{Mahmoud_Abdolrahim}:
Set $k_1=k_2=0$ and follow similar lines as in comparison with \cite{Dolk_LP_ieee,Dolk_Lp}, we get $t_{k+1}=\inf \{t\in\mathbb{R}:t>t_k+\tau_{m} \bigwedge  \varphi(\xi(t^-),\varepsilon(t^-))=0\}$. 
$\bullet$  \cite{Postuyan_GFW}: 
Let $k_1=0$, $\bar{\varphi}(t)=0$ for all $t\in\mathbb{R}$ and $\hat{s}_k={\phi}_2(\hat{t}_k^-)$. Choose $s_k\geq 0$ we have ${\phi}_2(t)\geq 0$ for all $t\geq t_0$. Then (\ref{trig_cond}), (\ref{phi}) reduce to $\varphi_2(\varepsilon)=-\varphi_1(\xi)+k_2{\phi_2}$, where $\dot{\phi}_2+\alpha_2({\phi_2})=0$. In this case, ${\phi_2}$ plays the role of threshold variable defined in \cite{Postuyan_GFW}. However, unlike the present work where ${\phi_2}$ appears in the TC as a positive term that is added to some functions of state's norm, in \cite{Postuyan_GFW}, the admissible measurement error is bounded by the maximum of these two. 
\subsection{$\mathcal{L}_p$-gain performance}\label{main section}
We start with a useful lemma which is an application of Lemma \ref{lem_lip} and its proof is given in the Appendix.
\begin{lem}\label{lem_error}
	Let $a=\lambda_1^p\psi(\hat{\tau},\lambda_2)$, $b=\lambda_3^p\psi(\hat{\tau},\lambda_2)$. Then
	\begin{eqnarray}
	\int_{t_k}^{\hat{t}_k}\|\varepsilon(\tau)\|^pd\tau
	\leq a\int_{t_k}^{\hat{t}_k}\|\xi(\tau)\|^pd\tau
	+ b\int_{t_k}^{\hat{t}_k}\|d(\tau)\|^pd\tau.\label{lem_error_eq}
	\end{eqnarray}
\end{lem}	
\begin{rmk}\label{rmk_a,b}
	In view of the definition of $a$, $b$ and $\hat{\tau}$, one can verify that $a\leq \lambda_1^p\psi(\tau_1,\lambda_2)$, $b\leq \lambda_3^p\psi(\tau_3,\lambda_2)$. Also from definition of $\tau_1$, $\tau_3$ we conclude $ac<B_1$, $bc\leq B_3$; inequalities that will be used later in the proof of main results.
\end{rmk}	
Next theorem states our primary result where the finite gain $\mathcal{L}_p$-stability of continuous-time system (\ref{eq sys_1}) is shown to be preserved under the event-based execution of control task. Compared to \cite{Dolk_LP_ieee,Mahmoud_Abdolrahim}, our result relies on a less conservative set of assumptions.
\begin{thm}\label{thm_main}
Under Assumptions \ref{assumption_ISS}, \ref{ass_alpha}, \ref{ass_impose} and ETM (\ref{trig_cond})-(\ref{tau_hat}) the ETS (\ref{eq sys}) is finite gain $\mathcal{L}_p$-stable with $\mathcal{L}_p$-gain $\leq \mu_d$. In addition, the origin $\xi=0$ is globally asymptotically stable.
\end{thm}
\begin{proof}
	For $t\in[t_k,\hat{t}_{k})$, Assumption \ref{assumption_ISS}(ii) suggests
	\begin{eqnarray}\label{W_dot}
	\dot{V}_c(\xi)\leq
	\mu^p\|d\|^p-\|z\|^p+\nabla V_c(\xi)\cdot g(\xi)(\gamma(\xi+\varepsilon)-\gamma(\xi))
	\end{eqnarray}
	which further reduces to
	\begin{eqnarray}
	\dot{V}_c(\xi) \leq \mu^p\|d\|^p-\|z\|^p+\lambda_2\| \nabla V_c(\xi)\|~\|\varepsilon\|\nonumber
	\end{eqnarray}
	by applying  $\|g(\xi)(\gamma(\xi+\varepsilon)-\gamma(\xi))\|\leq \lambda_2 \|\varepsilon\|$ (that is already proven in the proof of Lemma \ref{semi_global}). Thus, from Lemma \ref{lem_1}(ii) and Assumption \ref{assumption_ISS}(iii), we get
	\begin{eqnarray}
	\dot{V}_{c,\lambda(\xi)} \leq\frac{\lambda^q\bar{c}_3^q}{q}\|\xi\|^p
	+\frac{\lambda_2^p}{p}\|\varepsilon\|^p +\lambda \mu^p \|d\|^p -\lambda\|z\|^p.\nonumber
	\end{eqnarray}
	As a consequence, for  ${V}(\xi)=V_s(\xi)+V_{c,\lambda}(\xi)$ it follows from Assumption \ref{assumption_ISS}, (\ref{lem_error_eq}) and (\ref{W_dot}) that 
	\begin{eqnarray}
	{{V}}(\xi({\hat{t}_k})-{{V}}(\xi({t_k}))\leq -(c_1-\frac{\lambda^q\bar{c}_3^q}{q}-ac)\int_{{t_k}}^{{\hat{t}_k}}{\|\xi(\tau)\|^p}d\tau \nonumber \\ +(\lambda\mu^p+c_3+bc) \int_{{t_k}}^{{\hat{t}_k}}{\|d(\tau)\|^p}d\tau-\lambda\int_{{t_k}}^{{\hat{t}_k}}{\|z(\tau)\|^p}d\tau ~~\nonumber \\ 	\leq\lambda\mu_d^p \int_{{t_k}}^{{\hat{t}_k}}{\|d(\tau)\|^p}d\tau -\lambda\int_{{t_k}}^{{\hat{t}_k}}{\|z(\tau)\|^p}d\tau,~~~~~~~~~\nonumber
	\end{eqnarray} 
	where the last inequality follows from Remark \ref{rmk_a,b}.
	For $t\in[\hat{t}_k,t_{k+1})$ one can apply TC (\ref{trig_cond}) to obtain an upper bound on $\dot{{V}}$ as $\dot{{V}}(\xi) {\leq} -c_1(1-\sigma)\|\xi\|^p+(\lambda\mu^p{+}c_3)\|d\|^p-\lambda\|z\|^p+k_2{\phi_2}-\dot{\phi}_1$
	where $-\alpha_1({\phi_1})$ term is eliminated from the right hand side since $\phi_1$ is non-negative. It then follows that $\dot{{V}}(\xi) \leq \lambda\mu_d^p\|d\|^p-\lambda\|z\|^p+k_2\delta_k-\dot{\phi}_1$
	and hence
	\begin{eqnarray}
	{V}(\xi({t_{k+1}}))-{V}(\xi({\hat{t}_k}))\leq \lambda \mu_d^p \int_{{\hat{t}_k}}^{{t_{k+1}}}{\|d(\tau)\|^p}d\tau\nonumber \\   -\lambda\int_{{\hat{t}_k}}^{{t_{k+1}}}{\|z(\tau)\|^p}d\tau +k_2\int_{{\hat{t}_k}}^{{t_{k+1}}}{\delta_k(\tau)}d\tau+\hat{r}_k.\nonumber
	\end{eqnarray} 
	Therefore, we may conclude
	\begin{eqnarray}\label{recursive2}
	{V}(\xi({t_{k+1}}))-{V}(\xi({t_k}))\leq
	\lambda\mu_d^p\int_{{t_k}}^{{t_{k+1}}}{\|d(\tau)\|^p}d\tau \nonumber \\ -\lambda\int_{{t_k}}^{{t_{k+1}}}{\|z(\tau)\|^p}d\tau+k_2\int_{{\hat{t}_k}}^{{t_{k+1}}}{\delta_k(\tau)}d\tau+\hat{r}_k.\nonumber
	\end{eqnarray}
	Apply this inequality to the sampling intervals until $t\geq t_0$, the positive definiteness of ${V}$ can be employed to write
	\begin{eqnarray}
	\int_{t_0}^{t}{\|z(\tau)\|^p}d\tau \leq \mu_d^p \int_{t_0}^{t}{\|d(\tau)\|^p}d\tau+\frac{1}{\lambda}(k_2\theta_1+\theta_3+{V}(\xi_0)).\nonumber
	\end{eqnarray}
	This proves $\mathcal{L}_p$-stability of ETS (\ref{eq sys}) with $\mathcal{L}_p$-gain $\leq \mu_d$. To show asymptotic stability, let $d=0$. Using a similar process as we prove of $\mathcal{L}_p$-stability, it can be shown then suggests that for any for $\lambda=0$, any $t\geq t_0$, 
	\begin{eqnarray}
	{V_s} (\xi(t)) \leq-c_1(1-\sigma)\int_{t_0}^{t}{\|\xi(\tau)\|^p}d\tau+k_2 \theta_1+\theta_3+{V_s}(\xi_0).\nonumber
	\end{eqnarray}
	This proves the ultimate boundedness of trajectories of system (\ref{eq sys}). However, global asymptotic stability is postponed to show that for any $\epsilon\in\mathbb{R}^+$ there exists some $\delta\in\mathbb{R}^+$ such that if $\|\xi_0\|\leq\delta$, $\|\xi(t)\|\leq\epsilon$ for all $t\geq t_0$ and $\lim_{t\rightarrow\infty} \xi(t)=0$. This is achieved by redefining $\delta_k(t)$ (resp., $\hat{r}_k$) as $\lambda_0 {V_s}(\xi_0)\delta_k(t)$ (resp., $\lambda_0 {V_s}(\xi_0) \hat{r}_k$) for some $\lambda_0\in\mathbb{R}^+$. Thus by choosing 
	\begin{eqnarray}
	\delta={V}_s^{-1}\Big{(}\frac{{V_s}(\epsilon)}{1+\lambda_0(k_2 \theta_1+\theta_3)}\Big{)}\nonumber
	\end{eqnarray}
	for a given $\epsilon$, we have 
	\begin{eqnarray}
	{V_s}(\xi(t))\leq-c_1(1-\sigma)\int_{t_0}^{t}{\|\xi(\tau)\|^p}d\tau+{V_s}(\epsilon),\nonumber
	\end{eqnarray}
	{\it i.e.,} $\|\xi(t)\|\leq\epsilon$ for all $t\geq t_0$. Convergence of $\xi$ to zero is easy to show and omitted due to space limitations.
\end{proof}
\begin{rmk}\label{affine}
	If instead of affine structure (\ref{eq sys_0}), the system model is assumed to be $\dot{\xi}=f(\xi,u,d)$, our main findings which consist the results of Lemma \ref{semi_global} and Theorem \ref{thm_main} are still valid provided that $\varphi_3$ in (\ref{phi_3}) is replaced with $\varphi_3(r,s)  =\lambda_2 \nabla V_{c,\lambda}(r)\cdot \|s\|$. The details can be found in \cite{Me_submitted}.
\end{rmk}
\subsection{Inter-event time enlargement}
In the sequel, we present an important feature of TC (\ref{trig_cond}) on extending inter-event times. For this purpose, we define
\begin{eqnarray}
\tau^*_{max}\doteq\max \{\tau^*:\bar{\rho},\chi\in\mathbb{R}^+_0\},\nonumber
\end{eqnarray}
which in view of the following theorem, upper bounds the new extended inter-event times. In this definition $\tau^*$ is assumed to be a function of $\bar{\rho}$ and $\chi$, as suggested by (\ref{chi}) and dependence of $\lambda_i$, $i\in\{1,2,3\}$ on $\bar{\rho}$ (that is defined in Lemma \ref{lem_state_bound}).
\begin{thm}\label{thm_enlarge}
	For any $T^\circ\in\mathbb{R}^+$ and $\tau^\circ\in [0,\tau^*_{max}]$, $\bar{\varphi}$ in (\ref{Delta}) can be designed in a way that $t_{k+1}-t_k\geq \tau^\circ$ at least for $t_{k+1}\leq T^\circ$.
\end{thm}
\begin{proof}
	To find a lower bound on inter-event times, let us restrict the TC (\ref{trig_cond}) to $\varphi(\xi,\varepsilon)=k_2{\phi_2}$ by taking $k_1=0$. 
	Recalling the proof of Lemma \ref{semi_global} where the  triggering happens when $\chi=1$, our goal here is to design ${\phi_2}$ so that the triggering occurs for some $\chi>1$. 
	Note that $\tau^*_{max}\geq\tau^*(1)$ by definition.
	Due to continuity of $\tau^*$ in (\ref{chi}), for any $\tau^\circ\in [0,\tau^*_{max}]$ one can find $\chi^\circ$ (obviously $\geq1$) so that $\tau^\circ=\tau^*(\chi^\circ)$. It only remains to choose the TC such that $\chi\geq\chi^\circ$ at sampling instants. With the same notation as in Lemma \ref{semi_global}, let $\delta^*:={\chi^*}^2\bar{\delta}$ where $\chi^* = \chi^\circ +\frac{m_1} {m_2}{\bar{\rho}}(\chi^\circ-1)$. We redefine $\bar{\varphi}$ in (\ref{Delta}) as
	\begin{eqnarray}
	\bar{\varphi}(t)=\begin{cases}
	\alpha_2(\delta^*),~~~ t\in[0,T^\circ),\\
	0,~~~~~~~~~~\text{elsewhere}.
	\end{cases}\nonumber
	\end{eqnarray} 
	This implies ${\phi_2}(t)=\delta^*$ for $t\in[0,T^\circ)$. Then following similar lines as we derived (\ref{m1m2}), 
	the lower bound on the inter-event times can be calculated by assuming the triggering rule
	\begin{eqnarray}
	2\|\varepsilon\|=m_1{\|\xi\|}+\chi^*m_2.\nonumber
	\end{eqnarray}
	From definition of $\chi$ given in the proof of Lemma \ref{semi_global} it is easy to verify that
	\begin{eqnarray}
	\chi=\frac{m_1\|\xi\|+\chi^*m_2}{m_1\|\xi\|+m_2} \geq \chi^\circ\nonumber
	\end{eqnarray}
	at triggering instants and hence inter-event times are lower bounded by $\tau^\circ$ for $t\leq T^\circ$. 
\end{proof}
\begin{rmk}
	Theorem \ref{thm_enlarge} explores one of the advantages of our proposed strategy where the intersampling intervals are extended to $\tau^\circ$ for $t\in[0,T^\circ]$. The numerical example in section \ref{section examples} suggests that the average sampling time is also improved in this interval. Note that while the results are not explicitly applicable to $t>T^\circ$, numerical examples in \cite{Me_submitted} verify the efficiency of this technique for all $t\geq t_0$.
\end{rmk}
\section{Example}\label{section examples}
\subsection{System model (\cite{Me_submitted})}
Consider the system (\ref{eq sys_1}), with $\xi=[\xi_1~\xi_2]^T$ and\footnote{This example is of a Lur’e type system, which has recently seen attention in the context of event-based control, \cite{Lur's_statefeedback,Cone-nonlin}.}
\begin{eqnarray}\label{sys_exmp}
f(\xi,d)=\begin{pmatrix}
\xi_2\\-H(\xi_1)+d
\end{pmatrix},
~~g(\xi)=\begin{pmatrix}
0\\1
\end{pmatrix},
~~h(\xi,d)=\xi_1,
\nonumber
\end{eqnarray}
$u(t)=\gamma(\xi(t))= -\xi_2(t))$. The piecewise linear function $H:\mathbb{R}\mapsto\mathbb{R}$ is given by: $H(r)=2r$ for $|r|\leq h^*$, $H(r)=h^*+r$ for $ r\geq h^*$ and $H(r)=-h^*+r$ for $r\leq -h^*$, some $h^*\in\mathbb{R}_0^+$ for some non-negative $h^*$. Note that $H$ satisfies $r^2\leq rH(r)\leq 2r^2$ for any $r\in \mathbb{R}$. We study the finite gain $\mathcal{L}_2$-stability this system under event-based implementation of control law. 
 
\subsection{Verification of Assumption \ref{assumption_ISS}}
Letting $V_s(\xi)= \frac{\upsilon_1}{2}\xi^TP\xi+ 2\upsilon_1\int_{0}^{\xi_1}H(r)dr$, $P=[1~1;1~2]$, then (i) holds for $c_1=\frac{\upsilon_1}{2}$, $c_2=c_3=5\upsilon_1$ (see \cite{Me_submitted} for details). Here $\upsilon_1$ is the scaling factor discussed in Remark \ref{lambda}. To show (ii), we start with $\frac{1}{\upsilon_1}\dot{V}_s(\xi)=-\xi_1H(\xi_1)-\xi_2^2+(\xi_1+2\xi_2)d\leq -(1-n_1)\xi_1^2-(1-n_2)\xi_2^2+(\frac{1}{4n_1}+\frac{1}{n_2})d^2-n_1(\xi_1-\frac{d}{2\sigma})^2-n_2(\xi_2-\frac{d}{n_2})^2$ for some positive $n_1$, $n_2$. Choosing $V_c(x)=\frac{1}{\upsilon_1(1-n_2)}V_s(x)$ yields $\dot{V}_c(x)\leq |z|^2-\mu^2|d|^2$ where $\mu^2=\frac{1}{1-n_2}(\frac{1}{4n_1}+\frac{1}{n_2})$. The minimum  $\mu$ is $4.49$  obtained for $n_1=1$, $n_2=0.47$. A less conservative bound may be found with a different choice of $V_c$. 
Finally,  (iii) holds for $\bar{c}_1=\frac{\upsilon_1}{2}\lambda_{max}([5~1;1~2])$,  $\bar{c}_2=\frac{1}{\upsilon_1(1-n_2)}\bar{c}_1$ and $\bar{c}_3=\frac{1}{1-n_2}(\|P\|+4)$, where $\lambda_{max}(\cdot)$ is the maximum eigenvalue. 

\subsection{Triggering condition}
Our design criteria is to guarantee $\mu_d\leq 5$. We consider here two scenarios for $\delta_k$ in (\ref{Delta}):
\begin{eqnarray}
\delta^1_k(t)=D_1 e^{-\varrho_1 t},~~
\delta^2_k(t)=D_2 \frac{\varrho_2^n}{n!},~ n=\lceil \frac{t}{\bar{n}}\rceil,\nonumber
\end{eqnarray}
where $D_1=10$, $D_2=2$, $\varrho
_1=0.05$, $\varrho_2=3$, $\bar{n}=10$. Also, we consider $\alpha_1(r)=\alpha_2(r)=r$ in (\ref{phi}). To cover the strategies discussed in Section \ref{compare}, we categorize our analysis into six cases, depending on the values of $k_1$, $k_2$, $\delta_k^1$, $\delta_k^2$:
\begin{table}[H]
	\vspace*{-0.5em}
	\centering
	\label{}
	\renewcommand{\arraystretch}{.7}
	\begin{tabular}{ccccccc}
		\toprule[1.5pt]   
		case: & ~(i)&~(ii)& ~(iii)&  ~(iv)&~(v)&~(vi)\\
		\midrule
		$(k_1,k_2)$&$~(1,1)$&$~(1,1)$ & $~(1,0)$& $~(0,1)$& $~(0,1)$ & $~(0,0)$ \\
		\vspace{0em} $\delta_k$&$~\delta_k^1$&$~\delta_k^2$ & $~n/a$& $~\delta_k^1$& $~\delta_k^2$ & $~n/a$\\
		\bottomrule[1.5pt]
	\end{tabular}
\vspace{-0.5em}
\end{table}
Cases (i), (ii) are the general dynamic triggering scenarios with both ${\phi_1}$, ${\phi_2}$ effective in condition (\ref{trig_cond}). 
The role of $\phi_1$ (resp., $\phi_2$) is studied in case (iii) (resp., cases (iv), (v)). 
Also, case (vi) results in static TC since both ${\phi_1}$, ${\phi_2}$ are absent.

It is not difficult to verify $\lambda_1=3$, $\lambda_2=\lambda_3=1$ in Lemma \ref{lem_lip}. Therefore, we may choose $\lambda=4.7\times10^{-3}$, $\upsilon_1=3.6\times 10^{-3}$ (which satisfy the required bounds on $\lambda$ given in Remark \ref{lambda}) and obtain $\hat{\tau}=8.9\times10^{-3}$ from (\ref{tau_hat}). Finally, we take $\bar{\delta}=10$, $r_k=0$, $\hat{r}_k=\phi_1(\hat{t}_k^-)$, $s_k=12.5$.  
\subsection{Numerical simulation}
Signal $d(t)$ follows a zero mean Gaussian distribution with variance $1$ over $t\in[0,100)$ and zero everywhere else.
We also take $h^*=0.3$ and run the simulation for $100$ initial conditions uniformly distributed in a circle of radius $1$ over $100$ seconds and finally average the results. The plots are provided for initial condition $\xi_0=(\sin(\frac{\pi}{3}),\cos(\frac{\pi}{3}))$.
\begin{figure}[H]
	\vspace{-0.75em} 
	\hspace*{-0.6em}
	\centering
	\includegraphics[width=1.02\columnwidth]{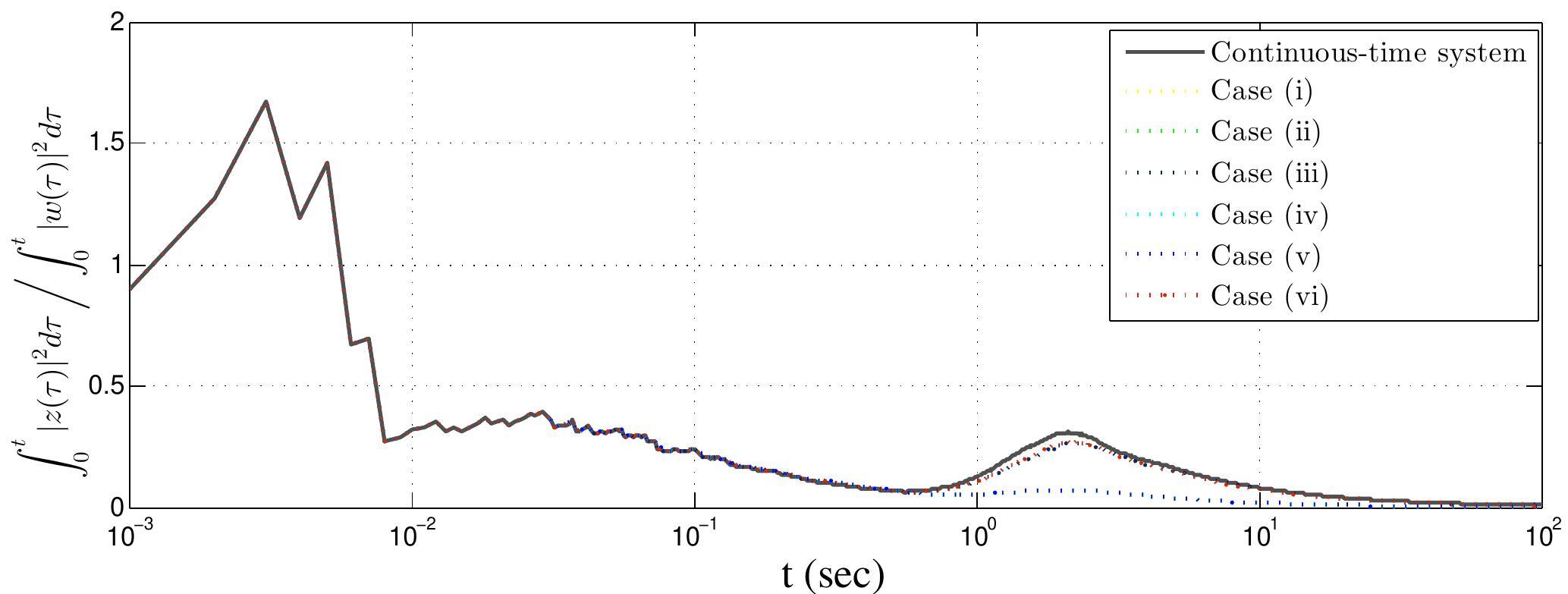}
		\vspace{-1.7em} 
	\caption{Verification of $\mathcal{L}_2$-gain.}
	\label{fig:ex2-1}
	\vspace{-1.5em}	
\end{figure}
\begin{figure}[H]
	\vspace{-0.5em} 
	\hspace*{-0.5em}
	\centering
	\includegraphics[width=1.02\columnwidth]{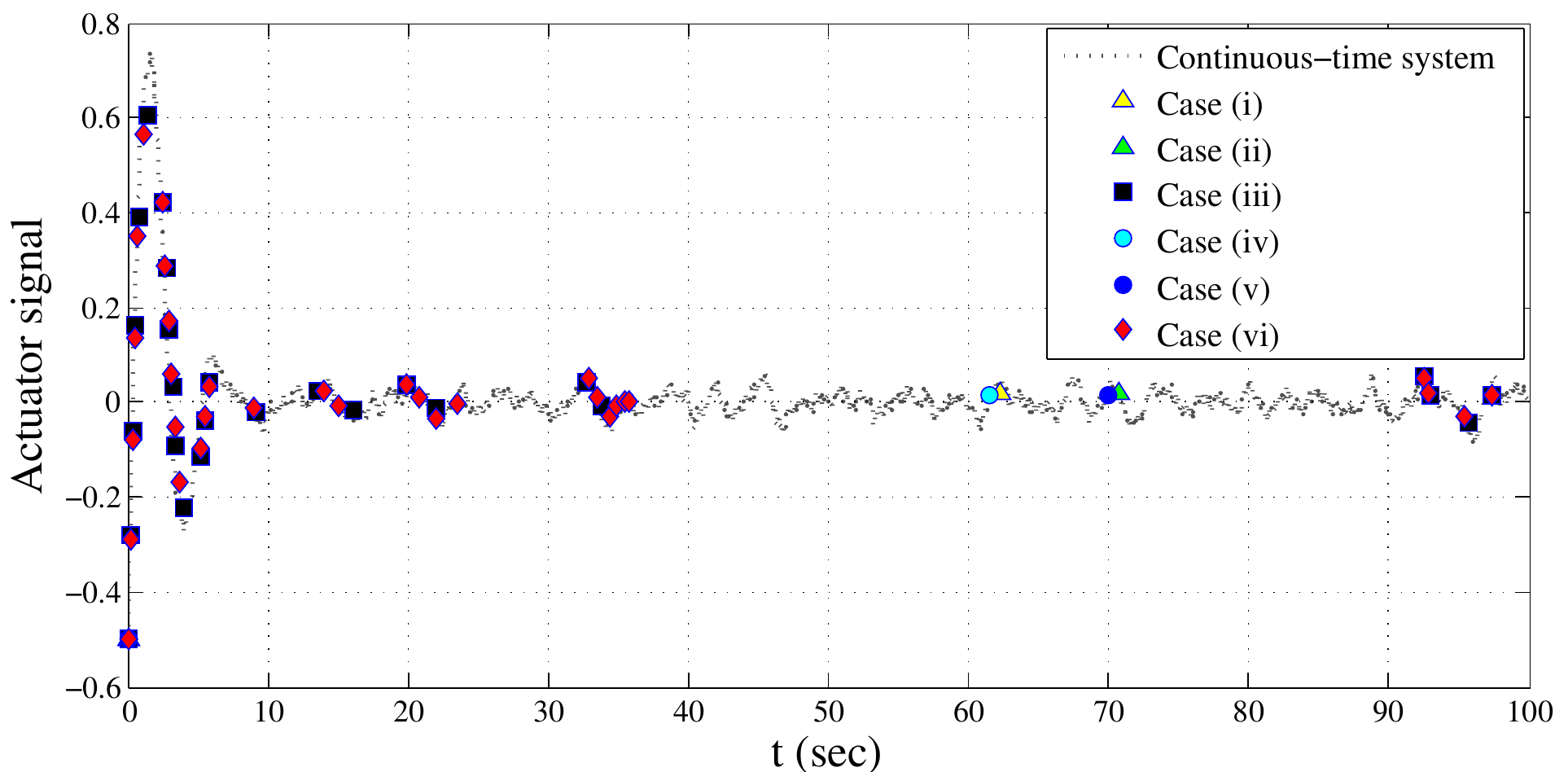}
	\vspace{-1.7em} 
	\caption{Actuator signal at the triggering instants}
	\label{fig:ex2-1}
	\vspace{-0.5em}	
\end{figure}
\vspace*{-0em}
\begin{table}[H]
	\vspace{-1.5em}
	\centering
	\caption{Comparison of different scenarios.}
	\label{tab:table_final}
	\renewcommand{\arraystretch}{.7}
	\vspace{-.5em}
	\begin{tabular}{ccccccc}
		\toprule[1.5pt]
		case:~~~~& (i)~~&(ii)~~& (iii)~~& (iv)~~&(v)~~&(vi)~\\
		\midrule
		$N$ ~~~~& 3.24~~&3.25~~ & 12.9~~ & 4.34~~& 4.72~~& 18.7~\\
		$\tau_{m}$$\times 10^{2}$ ~~~~& 22.3~~&14.2~~& 3.3~~ & 22.6~~& 14.8~~& 1.8~\\
		\bottomrule[1.5pt]
	\end{tabular}
\vspace{-0.5em}
\end{table}
Table \ref{tab:table_final} illustrates the number of triggerings ($N$) and MIET ($\tau_m$) for different scenarios. The values of $\tau_m$ are in msec. Comparing different cases, it is clear that both ${\phi_1}$, ${\phi_2}$ improve transmission rate, however, when $k_2$ is non-zero, the number of samples and $\tau_{m}$ improve more significantly. This implies the effectiveness of parameter ${\phi_2}$ compared to ${\phi_1}$. This example suggests that when trajectories of open-loop ETS are either converging to the origin or staying bounded, since $\varepsilon$ remains bounded, an appropriate choice of $\bar{\delta}$, $\delta_k$ in ${\phi_2}$ avoid unnecessary samplings effectively.    
The fact that $\tau_{m} \gg \hat{\tau}$ supports the merit of our design compared to time-regularization method as this latter scheme often degenerates to periodic samplings ($\hat{\tau}$ in the case of our design) when state is near origin. Our method indeed outperforms time-triggered scheme too (with a period equal to the MIET), since the average inter-event intervals, {\it i.e.,} $\frac{100}{N}$, are much larger than the MIETs.

\section{Conclusion}\label{section conclusion}
This paper introduces a framework for event-triggered design by focusing on $\mathcal{L}_p$ performance problem. Our proposed ETM is based on two dynamic variables ${\phi_1}$ and ${\phi_2}$. Indeed, ${\phi_1}$ has the role of dynamic TC (\cite{Girard}) and is intended to enlarge inter-event times. While ${\phi_2}$ serves to extend the inter-event times too, it also has the critical roles of (i) enabling us to analytically predict the increase of MIET for a desired period of time, and, more importantly, (ii) excluding Zeno behaviour.     
An interesting future research topic is to check if an \emph{output-based} TC can be equipped with the dynamic parameters $\phi_1,\phi_2$ to enjoy the benefits offered by these parameters. Note that contrary to dwell-time approaches where the Zeno-freeness is granted a priori, the output-based generalization of our approach requires Zeno behaviour to be carefully ruled out.   


\section*{Appendix}\label{section appendix}	
\begin{newproofof}
	\textit{Lemma \ref{lem_state_bound}.}
	We shall need the following proposition whose proof can be obtained applying integration by parts and Assumption \ref{ass_alpha}.
	\begin{proposition}\label{lem_phi-1}
		For ${\phi}_1$ defined in (\ref{phi}) and (\ref{IC}), we have 
		\begin{eqnarray}
		\int_{t_k}^{\hat{t}_k}-e^{\nu\tau} d{\phi}_1(\tau) \leq r_k e^{\nu t_k},~~\int_{\hat{t}_k}^{t_{k+1}}-e^{\nu\tau} d{\phi}_1(\tau) \leq \hat{r}_k e^{\nu \hat{t}_k}.\nonumber
		\end{eqnarray}
	\end{proposition} 
	Now we start from Assumption \ref{assumption_ISS}(ii) to write 
	\begin{eqnarray}
	~\dot{V}_c (\xi)\leq 
	\mu^p\|d\|^p -\|z\|^p+\nabla V_c(\xi)g(\xi) (\gamma(\xi+\varepsilon)-\gamma(\xi))\nonumber
	\end{eqnarray}
	Define ${V}(\xi)=V_s(\xi)+V_{c,\lambda}(\xi)$, one can apply Assumption \ref{assumption_ISS}(i), (\ref{trig_cond}), (\ref{varphi}) to get
	\begin{eqnarray}
	\dot{{V}}(\xi)\leq -c_1(1-\sigma)\|\xi\|^p+(\lambda\mu^p+c_3)\|d\|^p+\varphi\leq ~~~~~~~~~~~~~~~~~ \nonumber\\  -c_1(1-\sigma)\|\xi\|^p+(\lambda\mu^p+c_3)\|d\|^p+k_2{\phi_2}-\alpha_1({\phi_1})-\dot{\phi}_1~~\nonumber
	\end{eqnarray}
	where the second inequality is obtained using (\ref{phi}). Let $A={\lambda\mu_d^p{\|d\|}_\infty^p+ {k_2\|{\phi_2}\|}_{\infty}}$,
	we can apply Proposition \ref{prop_phi_positive} to write 
	\begin{eqnarray}
	\dot{{V}} (\xi)+ \nu{V}(\xi)\leq A-\dot{\phi}_1,\nonumber
	\end{eqnarray}
	where we used ${V}(\xi)\leq(\bar{c}_1+\bar{c}_2)\|\xi\|^p$ suggested by Assumption \ref{assumption_ISS}(iii). 
	We then conclude from Proposition \ref{lem_phi-1} that 
	\begin{eqnarray}
	{V}(\xi(\hat{t}_k)) e^{\nu{{\hat{t}_k}}}\leq {V}(\xi(t_k))e^{\nu{{t}_k}}+r_k e^{\nu t_k} +A\int_{t_k}^{\hat{t}_k}e^{\nu\tau}d\tau,\nonumber~~~~ \\
	{V}(\xi(t_{k+1})) e^{\nu{t_{k+1}}}	\leq {V}(\xi(\hat{t}_k))e^{\nu{\hat{t}_k}}+\hat{r}_k e^{\nu \hat{t}_k}+A\int_{\hat{t}_k}^{t_{k+1}}e^{\nu\tau}d\tau.\nonumber
	\end{eqnarray}
	Adding the two inequalities and apply the result to the sampling intervals until $t\geq t_0$, Assumption \ref{ass_impose} yields ${V}(\xi(t))e^{\nu t}\leq {V}(\xi_0)+(\theta_2+\theta_3)e^{\nu t}+A\int_{t_0}^{t}e^{\nu\tau}d\tau$. Therefore, 
	\begin{eqnarray}
	{V}(\xi(t)) \leq  {V}(\xi_0) e^{-\nu t}+\theta_2+\theta_3+A\int_{t_0}^{t}e^{-\nu(t-\tau)}d\tau \nonumber\\ \leq  {V}(\xi_0)+\theta_2+\theta_3+\nu^{-1}A~~~\quad~~~~~~~~~~~~~\nonumber
	\end{eqnarray}
	which gives the desired result.
\end{newproofof}

\begin{newproofof}
	\textit{Proposition \ref{prop_phi_positive}.}
	From (\ref{trig_cond}), (\ref{phi}) $\phi_1$ satisfies $\dot{\phi}_1 + \alpha_1({\phi_1}) + k_1{\phi_1}\geq 0$ for  $t\in[t_k,t_{k+1})$. Note that ${\phi_1}(t)\equiv 0$ is a solution to $\dot{\phi}_1+\alpha_1({\phi_1})+k_1{\phi_1}= 0$. Therefore, since ${\phi_1}(t_k),{\phi_1}(\hat{t}_k)\geq 0$ it follows that ${\phi_1}(t)\geq 0$ for all $t\geq t_0$. For the second part, since $\bar{\delta}$ (resp., $\delta_k(t)$) is a solution of $\phi_2$ in (\ref{phi}) for $t\in[t_k,\hat{t}_k)$ (resp., $t\in[\hat{t}_k,t_{k+1})$) and $s_k\geq \bar{\delta}$ (resp., $\hat{s}_k=\delta_k(\hat{t}_k)$), it follows that ${\phi}_2(t)\geq \bar{\delta}$ (resp., ${\phi}_2(t)=\delta_k(t)$) over this interval. Finally, from the positiveness of $\bar{\delta}$ and $\delta_k(t)$, $\phi_2(t)\geq 0$ for all $t\geq t_0$. 
\end{newproofof}

\begin{newproofof}
	\textit{Lemma \ref{lem_error}.}
	We first define the notation below:
	\begin{eqnarray}
	\mathcal{I}(x)=\int_{t_k}^{s}e^{\lambda_2 (s-\tau)}\|x(\tau)\| d\tau,~
	\mathcal{Q}(s)=(\int_{t_k}^{s}e^{\frac{\lambda_2 q}{2} (s-\tau)} d\tau)^{\frac{p}{q}},\nonumber\\
	\mathcal{J}(x)=\int_{t_k}^{s}e^{\frac{\lambda_2 p}{2} (s-\tau)}\|x(\tau)\|^p d\tau,~\mathcal{P}(s)=\int_{t_k}^{s}e^{\frac{\lambda_2 p}{2} (s-t_k)} ds.\nonumber
	\end{eqnarray}
	From definition of $\varepsilon$ and Lemma \ref{lem_lip} we have
	\begin{eqnarray}
	\frac{d\|\varepsilon\|}{dt} \leq \|\dot{\varepsilon}\|=\|\dot{\xi}\|\leq \lambda_1\|\xi\|+\lambda_2\|\varepsilon\|+\lambda_3\|d\|,\nonumber
	\end{eqnarray}
	solving which for $\varepsilon(t_k)=0$ and $s\geq t_k$ gives $\|\varepsilon(s)\|\leq \lambda_1 \mathcal{I}(\xi)+\lambda_3 \mathcal{I}(d)$.
	Then from Lemma \ref{lem_1}(i) we conclude 
	\begin{eqnarray}
	\|\varepsilon(s)\|^p \leq 2^p(\lambda_1^p {\mathcal{I}}^p(\xi){+} \lambda_3^p {\mathcal{I}}^p(d) ) {\leq} 2^p{\mathcal{Q}}(s) (\lambda_1^p{\mathcal{J}}(\xi){+}\lambda_3^p {\mathcal{J}}(d) )\nonumber
	\end{eqnarray}
	where the last inequality is obtained using Lemma \ref{lem_1}(ii).
It is then straightforward to check that for $t \geq t_k$, $\int_{t_k}^{t}{\mathcal{J}}(\xi) ds\leq {\mathcal{P}}(t) \int_{t_k}^{t}\|\xi(\tau)\|^p d\tau$ 
and hence conclude
\begin{eqnarray}
\int_{t_k}^{t}\|\varepsilon(s)\|^pds \leq 2^p\int_{t_k}^{t} {\mathcal{Q}}(s)(\lambda_1^p{\mathcal{J}}(\xi)+\lambda_3^p{\mathcal{J}}(d))ds ~~~~~~~~~ \nonumber \\
~~~~~~\leq 2^p{\mathcal{Q}}(t){\mathcal{P}}(t)(\lambda_1^p\int_{t_k}^{t}\|\xi(\tau)\|^p d\tau+\lambda_3^p\int_{t_k}^{t}\|d(\tau)\|^p d\tau). \nonumber
\end{eqnarray}
The proof is then complete taking $t=\hat{t}_k$ since $\psi(\hat{\tau},\lambda_2)=2^p{\mathcal{Q}}(\hat{t}_k){\mathcal{P}}(\hat{t}_k)$.
\end{newproofof}

\bibliographystyle{IEEEtranS}
\bibliography{GF}

\begin{thebibliography}{10}
\providecommand{\url}[1]{#1}
\csname url@samestyle\endcsname
\providecommand{\newblock}{\relax}
\providecommand{\bibinfo}[2]{#2}
\providecommand{\BIBentrySTDinterwordspacing}{\spaceskip=0pt\relax}
\providecommand{\BIBentryALTinterwordstretchfactor}{4}
\providecommand{\BIBentryALTinterwordspacing}{\spaceskip=\fontdimen2\font plus
\BIBentryALTinterwordstretchfactor\fontdimen3\font minus
  \fontdimen4\font\relax}
\providecommand{\BIBforeignlanguage}[2]{{%
\expandafter\ifx\csname l@#1\endcsname\relax
\typeout{** WARNING: IEEEtranS.bst: No hyphenation pattern has been}%
\typeout{** loaded for the language `#1'. Using the pattern for}%
\typeout{** the default language instead.}%
\else
\language=\csname l@#1\endcsname
\fi
#2}}
\providecommand{\BIBdecl}{\relax}
\BIBdecl

\bibitem{Mahmoud_Abdolrahim}
M.~Abdelrahim, J.~Daafouz, and D.~Ne{\v{s}}i{\'{c}}, ``Robust event-triggered
  output feedback controllers for nonlinear systems,'' \emph{Automatica},
  vol.~75, pp. 96--108, 2017.

\bibitem{PID}
K.~E. Arzen, ``A simple event-based pid controller,'' \emph{Proc. IFAC World
  Conf.}, vol.~18, pp. 423--428, 1999.

\bibitem{astrom_stochastic}
K.~Astrom and B.~Bernhardsson, ``Comparison of periodic and event based
  sampling for first order stochastic systems,'' \emph{Proc. IFAC World Conf.},
  pp. 301--306, 1999.

\bibitem{event-separation}
D.~P. Borgers and W.~P. M.~H. Heemels, ``Event-separation properties of
  event-triggered control systems,'' \emph{IEEE Trans. Autom. Control},
  vol.~59, no.~10, pp. 2644--2656, 2014.

\bibitem{Dolk_Lp}
V.~S. Dolk, D.~P. Borgers, and W.~P. M.~H. Heemels, ``Dynamic event-triggered
  control: Tradeoffs between transmission intervals and performance,''
  \emph{Proc. IEEE Conf. Dec. Control}, pp. 2764 -- 2769, 2014.

\bibitem{Dolk_LP_ieee}
------, ``Output-based and decentralized dynamic event-triggered control with
  guaranteed $\mathcal{L}_p$-gain performance and zeno-freeness,'' \emph{IEEE
  Trans. Autom. Control}, vol.~62, no.~1, pp. 34--49, 2017.

\bibitem{L_infty_gain}
M.~Donkers and W.~Heemels, ``Output-based event-triggered control with
  guaranteed $\mathcal{L}_\infty$-gain and improved and decentralized
  event-triggering,'' \emph{IEEE Trans. Autom. Control}, vol.~57, no.~6, pp.
  1362--1376, 2012.

\bibitem{Forni}
F.~Forni, S.~Galeani, D.~Ne{\v{s}}i{\'{c}}, and L.~Zaccarian, ``Event-triggered
  transmission for linear control over communication channels,''
  \emph{Automatica}, vol.~50, no.~2, pp. 490--498, 2014.

\bibitem{Me_submitted}
M.~Ghodrat and H.~J. Marquez, ``On the local input-output stability of
  event-triggered control systems,'' \emph{IEEE Trans. Autom. Control},
  vol.~64, no.~1, pp. 174--189, 2019.

\bibitem{Girard}
A.~Girard, ``Dynamic triggering mechanisms for event-triggered control,''
  \emph{IEEE Trans. Autom. Control}, vol.~60, no.~7, pp. 1992 -- 1997, 2015.

\bibitem{Heemels_periodic}
W.~Heemels, M.~Donkers, and A.~R. Teel, ``Periodic event-triggered control for
  linear systems,'' \emph{IEEE Trans. Autom. Control}, vol.~58, no.~4, pp.
  847--861, 2013.

\bibitem{lemmonlinearL_2fullstate}
M.~Lemmon, T.~Chantem, X.~Hu, and M.~Zyskowski, ``On self-triggered full
  information h-infinity\ controllers,'' \emph{Hybrid Systems: Computation and
  Control}, July 2007.

\bibitem{Lunze}
J.~Lunze and D.~Lehmann, ``A state-feedback approach to event-based control,''
  \emph{Automatica}, vol.~46, no.~1, pp. 211--215, 2010.

\bibitem{decentralized-Mazo}
M.~Mazo and M.~Cao, ``Asynchronous decentralized event-triggered control,''
  \emph{Automatica}, vol.~50, no.~12, pp. 3197--3203, 2014.

\bibitem{Mazo_Tabuada}
M.~Mazo and P.~Tabuada, ``Decentralized event-triggered control over wireless
  sensor/actuator networks,'' \emph{IEEE Trans. Autom. Control}, vol.~56,
  no.~10, pp. 2456--2461, 2011.

\bibitem{me_iet}
S.~H. Mousavi, M.~Ghodrat, and H.~J. Marquez, ``Integral-based event-triggered
  control scheme for a general class of non-linear systems,'' \emph{IET Control
  Theory Appl.}, vol.~9, no.~13, pp. 1982--1988, 2015.

\bibitem{explicit-nesic}
D.~Ne{\v{s}}i{\'{c}}, A.~R. Teel, and D.~Carnevale, ``Explicit computation of
  the sampling period in emulation of controllers for nonlinear sampled-data
  systems,'' \emph{IEEE Trans. Autom. Control}, vol.~54, no.~3, pp. 619--624,
  2009.

\bibitem{Postuyan_GFW}
R.~Postoyan, P.~Tabuada, D.~Ne{\v{s}}i{\'{c}}, and A.~Anta, ``A framework for
  the event-triggered stabilization of nonlinear systems,'' \emph{IEEE Trans.
  Autom. Control}, vol.~60, no.~4, pp. 982--996, 2015.

\bibitem{Seyboth}
G.~S. Seyboth, D.~V. Dimarogonas, and K.~H. Johansson, ``Event-based
  broadcasting for multi-agent average consensus,'' \emph{Automatica}, vol.~49,
  no.~1, pp. 245--252, 2013.

\bibitem{tabuada}
P.~Tabuada, ``Event-triggered real-time scheduling of stabilizing control
  task,'' \emph{IEEE Trans. Autom. Control}, vol.~52, no.~9, pp. 1680 -- 1685,
  2007.

\bibitem{Chopra}
P.~Tallapragada and N.~Chopra, ``Event-triggered dynamic output feedback
  control for \uppercase{LTI} systems,'' \emph{Proc. IEEE Conf. Dec. Control},
  pp. 6597--6602, 2012.

\bibitem{Cone-nonlin}
S.~Tarbouriech, A.~Seuret, L.~G. Moreira, and J.~M.~G. da~Silva,
  ``Observer-based event-triggered control for linear systems subject to
  cone-bounded nonlinearities,'' \emph{IFAC-PapersOnLine}, vol.~50, no.~1, pp.
  7893--7898, 2017.

\bibitem{L_2_selftriggered}
X.~Wang and M.~Lemmon, ``Self-triggered feedback control systems with
  finite-gain $\mathcal{L}_2$ stability,'' \emph{IEEE Trans. Autom. Control},
  vol.~54, no.~3, pp. 452--467, 2009.

\bibitem{Wang-distributed}
------, ``Event-triggering in distributed networked control systems,''
  \emph{IEEE Trans. Autom. Control}, vol.~56, no.~3, pp. 586--601, 2011.

\bibitem{passive_delay}
H.~Yu and P.~Antsaklis, ``Event-triggered output feedback control for networked
  control systems using passivity: Time-varying network induced delays,''
  \emph{Proc. IEEE Conf. Dec. Control and Eur. Control Conf.}, pp. 205--210,
  December 2011.

\bibitem{passive_input_output}
------, ``Event-triggered output feedback control for networked control systems
  using passivity: Triggering condition and limitations,'' \emph{Proc. IEEE
  Conf. Dec. Control and Eur. Control Conf.}, pp. 199--204, December 2011.

\bibitem{Lur's_statefeedback}
F.~Zhang, M.~Mazo, and N.~van~de Wouw, ``Absolute stabilization of lur’e
  systems under event-triggered feedback,'' \emph{IFAC-PapersOnLine}, vol.~50,
  no.~1, pp. 15\,301--15\,306, 2017.

\end{thebibliography}

\end{document}